\documentclass{article}
\usepackage{latexsym,amsmath,amssymb,url,amsthm,url}

\newcommand{\Prob}{{\mathbb P}}
\newcommand{\Exp}{{\mathbb E}}

\newtheorem{corollary}{Corollary}

\newtheorem{theorem}{Theorem}

\newtheorem{conjecture}{Conjecture}

\begin{document}

\title{Hypercontractive Inequality for Pseudo-Boolean Functions of Bounded Fourier Width}

\author{Gregory Gutin\\Royal Holloway, University of London\\
  Egham, Surrey TW20 0EX\\ United Kingdom, \url{gutin@cs.rhul.ac.uk} \and  Anders Yeo\\University of Johannesburg\\PO Box 524
Auckland Park 2006\\
South Africa, \url{andersyeo@gmail.com}}

\maketitle

\begin{abstract}
\noindent
 A function $f:\ \{-1,1\}^n\rightarrow \mathbb{R}$ is called pseudo-Boolean. It is well-known that each pseudo-Boolean function $f$ can be written as
$f(x)=\sum_{I\in {\cal F}}\hat{f}(I)\chi_I(x),$
where ${\cal F}\subseteq \{I:\  I\subseteq [n]\}$, $[n]=\{1,2,\ldots ,n\}$, and $\chi_I(x)=\prod_{i\in I}x_i$ and $\hat{f}(I)$ are non-zero reals.
The degree of $f$ is $\max \{|I|:\ I\in {\cal F}\}$ and the width of $f$ is the minimum integer $\rho$ such that every  $i\in [n]$ appears in at most $\rho$
sets in $\cal F$.
For $i\in [n]$, let $\mathbf{x}_i$ be a random variable taking values $1$ or $-1$ uniformly and independently from all other variables $\mathbf{x}_j$, $j\neq i.$ Let $\mathbf{x}=(\mathbf{x}_1,\ldots ,\mathbf{x}_n)$. The $p$-norm of $f$ is $||f||_p=(\mathbb E[|f(\mathbf{x})|^p])^{1/p}$ for any $p\ge 1$.
It is well-known that $||f||_q\ge ||f||_p$ whenever $q> p\ge 1$. However, the higher norm can be bounded by the lower norm times a coefficient not directly depending on $f$: if $f$ is of degree $d$ and $q> p>1$ then
$ ||f||_q\le \left(\frac{q-1}{p-1}\right)^{d/2}||f||_p.$
This inequality is called the Hypercontractive Inequality. We show that one can replace $d$ by $\rho$ in the Hypercontractive Inequality for each $q> p\ge 2$ as follows: $ ||f||_q\le ((2r)!\rho^{r-1})^{1/(2r)}||f||_p,$ where $r=\lceil q/2\rceil$. For the case $q=4$ and $p=2$, which is important in many applications, we prove a stronger inequality: $ ||f||_4\le (2\rho+1)^{1/4}||f||_2.$
\end{abstract}


\section{Introduction}\label{section:intro}

Fourier analysis of {\em pseudo-Boolean functions}\footnote{Often functions $f:\ \{0,1\}^n\rightarrow \mathbb{R}$ are called pseudo-Boolean \cite{BH02}.
In Fourier Analysis, the Boolean domain is often assumed to be $\{-1,1\}^n$ rather than the more usual $\{0,1\}^n$
and we will follow this assumption in our paper.}, i.e., functions $f:\ \{-1,1\}^n\rightarrow \mathbb{R}$,
has been used in many areas of computer science  (cf. \cite{AloGutKimSzeYeo11,CroGutJonKimRuz10,HasVen04,odonn,deWolf}), social choice theory (cf. \cite{FriKalNao2002,KahKalLin1988,MosOdoOle2010}), combinatorics,
learning theory, coding theory, and many others (cf. \cite{odonn,deWolf}).
We will use the following well-known and easy to prove fact \cite{odonn}:
each function $f:\ \{-1,1\}^n\rightarrow \mathbb{R}$ can be uniquely written as
\begin{equation}\label{eq1a}f(x)=\sum_{I\in {\cal F}}\hat{f}(I)\chi_I(x),\end{equation}
where ${\cal F}\subseteq \{I:\  I\subseteq [n]\}$, $[n]=\{1,2,\ldots ,n\}$, and $\chi_I(x)=\prod_{i\in I}x_i$ and $\hat{f}(I)$ are non-zero reals. Formula (\ref{eq1a}) is the Fourier expansion of $f$ and $\hat{f}(I)$ are the Fourier coefficients of $f$.
The right hand size of (\ref{eq1a}) is a polynomial and the degree $\max \{|I|:\ I\in {\cal F}\}$ of this polynomial will be called the {\em degree} of $f$.
For $i\in [n]$, let $\rho_i$  be the number of sets $I\in {\cal F}$ such that $i\in I$. Let us call $\rho=\max\{\rho_i:\ i\in [n]\}$ the {\em Fourier width} (or, just {\em width}) of $f$. The Fourier width was introduced in \cite{GutKimSzeYeo11} without giving it a name.

The degree and width can be viewed as dual parameters in the following sense. Consider a bipartite graph $G$ with partite sets $V$ and $T$, where $V$ is the set of variables in $f$ and $T$ is the set of terms in $f$ in (\ref{eq1a}), and $zt$ is an edge in $G$ if $z$ is a variable in $t\in T.$ Note that the degree of $f$ is the maximum degree of a vertex in $T$ and the width of $f$ is the maximum degree of a vertex in $V$.

For $i\in [n]$, let $\mathbf{x}_i$ be a random variable taking values $1$ or $-1$ uniformly and independently from all other variables $\mathbf{x}_j$, $j\neq i.$ Let $\mathbf{x}=(\mathbf{x}_1,\ldots ,\mathbf{x}_n)$. Then $f(\mathbf{x})$ is a random variable and  the $p$-{\em norm} of $f$ is $||f||_p=(\mathbb E[|f(\mathbf{x})|^p])^{1/p}$ for any $p\ge 1$. It is easy to show that $||f||_2^2=\sum_{I\in {\cal F}}\hat{f}(I)^2,$
which is Parseval's Identity for pseudo-Boolean functions.
It is well-known and easy to show that $||f||_q\ge ||f||_p$ whenever $q\ge p\ge 1$. However, the higher norm can be bounded by the lower norm times a coefficient not depending on $f$: if $f$ is of degree $d$ then
\begin{equation}\label{hyperc} ||f||_q\le \left(\frac{q-1}{p-1}\right)^{d/2}||f||_p.\end{equation}
The last inequality is called the {\em Hypercontractive Inequality}. (In fact, the Hypercontractive Inequality is often
stated differently, but the Hypercontractive Inequality in the original form and (\ref{hyperc}) are equivalent.)
Since $||f||_2$ is easy to
compute, the Hypercontractive Inequality is quite useful for $p=2$ and is often used for $p=2$
and $q=4$; this special case of the Hypercontractive
Inequality has been applied in many papers on algorithmics, social choice theory and many other areas, see, e.g.,
\cite{AloGutKimSzeYeo11,AloGutKri04,FriKalNao2002,GutIerMniYeo,GutKimSzeYeo11,HasVen04,KahKalLin1988,MosOdoOle2010}
and was given special proofs (cf. \cite{FriRod2001} and the extended abstract of \cite{MosOdoOle2010}).
We will call this case the {\em (4,2)-Hypercontractive Inequality}.

Theorem \ref{42in} proved below replaces the coefficient $3^{d/2}$ before $||f||_2$ in the (4,2)-Hypercontractive Inequality by
$(2\rho+1-\frac{2\rho}{m})^{1/4},$ where $\rho$ is the width of $f$ and $m=|{\cal F}|$. For functions with $2\rho+1<9^d$
Theorem \ref{42in} provides an important special case of the Hypercontractive Inequality with a smaller coefficient.
Note that in some cases one can change variables (using a different basis) such that the degree of $f$  decreases significantly.
However, this is not always possible and, even if it is possible, it might be hard to find an appropriate basis. Our
application of Theorem \ref{42in} in Section \ref{sec:fr} provides a nontrivial illustration of such a situation.
Note that Theorem \ref{42in} improves Lemma 7 in \cite{GutKimSzeYeo11}. While in Lemma 7 \cite{GutKimSzeYeo11},
the coefficient before $||f||_2$ is $(2\rho^2)^{1/4}$ ($\rho\ge 2$), in Theorem \ref{42in}, we decrease it to $(2\rho+1-\frac{2\rho}{m})^{1/4}.$
We provide examples showing that this coefficient is tight.

Due to Theorem \ref{42in}, we know that the width can replace the degree as a parameter in the coefficient before
$||f||_2$ in the (4,2)-Hypercontractive Inequality. A natural question is whether the same is true in the general case of
the Hypercontractive Inequality for pseudo-Boolean functions. We show that we can replace $d$ by $\rho$
for each $q\ge p\ge 2$ as follows: $ ||f||_q\le ((2r)!\rho^{r-1})^{1/(2r)}||f||_p,$ where $r=\lceil q/2\rceil$.

\section{(4,2)-Hypercontractive Inequality}

In (\ref{eq1a}), let ${\cal F}=\{I_1,\ldots, I_m\}$, $f_j(x)=\hat{f}(I_j)\chi_{I_j}(x)$ and $w_j=\hat{f}(I_j)$, $j\in [m].$ If $\emptyset \in {\cal F}$, we will assume that $I_1=\emptyset.$

\begin{theorem}\label{42in}
Let $f(x)$ be a pseudo-Boolean function of width $\rho\ge 0$. Then \newline $||f||_4\le (2\rho+1-\frac{2\rho}{m})^{1/4}||f||_2$.
\end{theorem}
\begin{proof}
If $\rho=0$ then $f(x)=c$, where $c$ is a constant and hence $||f||_4=||f||_2=c$. Thus, assume that $\rho\ge 1.$
Let $S$ be the set of quadruples $(p_1,p_2,p_3,p_4)\in [m]^4$ such that $\sum_{j=1}^4|\{i\}\cap I_{p_j}|$ is even for each $i\in [n]$,
$S'=\{(p_1,p_2,p_3,p_4)\in S:\ p_1=p_2\}$ and $S''=S\setminus S'.$ Note that if a product $f_p(\mathbf{x})f_q(\mathbf{x})f_s(\mathbf{x})f_t(\mathbf{x})$ contains a variable $\mathbf{x}_i$ in only one or three of the factors, then $\mathbb{E}[f_p(\mathbf{x})f_q(\mathbf{x})f_s(\mathbf{x})f_t(\mathbf{x})]=\mathbb{E}[P]\cdot \mathbb{E}(\mathbf{x}_i)=0,$
where $P$ is a polynomial in random variables $\mathbf{x}_l$, $l\in [n]\setminus \{i\}.$ Thus, $$\mathbb{E}[f(\mathbf{x})^4]=\sum_{(p,q,s,t)\in S}\mathbb{E}[f_p(\mathbf{x})f_q(\mathbf{x})f_s(\mathbf{x})f_t(\mathbf{x})].$$

Observe that if $(p,q,s,t)\in S'$ then $p=q$ and $s=t$ and, thus, \newline $\sum_{(p,q,s,t)\in S'}\mathbb{E}[f_p(\mathbf{x})f_q(\mathbf{x})f_s(\mathbf{x})f_t(\mathbf{x})]=\sum_{p=1}^m\sum_{s=1}^m w^2_pw_s^2.$
For a pair $(p,q)\in [m]^2$, let $N(p,q)=|\{(s,t)\in [m]^2:\ (p,q,s,t)\in S''\}|.$ Let a quadruple $(p,q,s,t)\in S''$.
Since $p\neq q$, there must be an $i$ which belongs to just one of the two sets $I_p$ and $I_q$.
Since $(p,q,s,t)\in S''$, $i$ must also belong to just one of the two sets $I_s$ and $I_t$ (two choices).
Assume that $i\in I_s$. Then by the definition of $\rho$, $s$ can be chosen from a subset of $[m]$ of cardinality
at most $\rho$. Once $s$ is chosen, there is a unique choice for $t$. Therefore, $N(p,q)\le 2\rho.$

Note that $(p,q,s,t)\in S''$ if and only if
$(s,t,p,q)\in S''$  which implies that there are at most $N(p,q)$ tuples in $S''$ of the
form $(s,t,p,q)$. We also have $$\mathbb{E}[f_p(\mathbf{x})f_q(\mathbf{x})f_s(\mathbf{x})f_t(\mathbf{x})]\le w_pw_qw_sw_t\le
(w_p^2w_q^2+w_s^2w_t^2)/2.$$ Thus, $$\sum_{(p,q,s,t)\in S''}\mathbb{E}[f_p(\mathbf{x})f_q(\mathbf{x})f_s(\mathbf{x})f_t(\mathbf{x})]\le \sum_{1\le p\neq q\le m}2N(p,q)\frac{w^2_pw^2_q}{2}\le 2\rho \sum_{1\le p\neq q\le m}w^2_pw^2_q.$$ Hence, $$\mathbb{E}[f(\mathbf{x})^4]\le \sum_{p=1}^m\sum_{s=1}^m w^2_pw_s^2 + 2\rho \sum_{1\le p\neq q\le m}w^2_pw^2_q=(2\rho+1) \sum_{p=1}^m\sum_{s=1}^m w^2_pw_s^2 - 2\rho \sum_{p=1}^mw^4_p.$$
We have $$\frac{\sum_{p=1}^mw^4_p}{\sum_{p=1}^m\sum_{s=1}^m w^2_pw_s^2}\ge \frac{\sum_{p=1}^mw^4_p}{\sum_{p=1}^m\sum_{s=1}^m [w^4_p/2+w_s^4/2]}=\frac{\sum_{p=1}^mw^4_p}{m \sum_{p=1}^mw^4_p}=\frac{1}{m}.$$
Thus, $\mathbb{E}[f(\mathbf{x})^4]\le (2\rho+1- \frac{2\rho}{m}) \left[\sum_{i=1}^mw_i^2\right]^2= (2\rho+1- \frac{2\rho}{m})\mathbb{E}[f(\mathbf{x})^2]^2.$ The last equality follows from Parseval's Identity.
\end{proof}

The following two examples show the sharpness of this theorem.

Let $f(x)=1+\sum_{i=1}^nx_i$. By Parseval's Indentity, $\mathbb{E}[f(\mathbf{x})^2]=n+1.$ It is easy to check that
$\mathbb{E}[f(\mathbf{x})^4]=(n+1)+{4 \choose 2}{n+1\choose 2}=3n^2+4n+1$. Clearly, $\rho=1$ and $m=n+1$ and, thus, $2\rho+1- \frac{2\rho}{m}=3-\frac{2}{n+1}.$
Also, $\mathbb{E}[f(\mathbf{x})^4]/\mathbb{E}[f(\mathbf{x})^2]^2=\frac{3n^2+4n+1}{(n+1)^2}=3-\frac{2}{n+1}.$

Let $f(x)=\sum_{I\subseteq [n]}\chi_I(x).$ Clearly, $\mathbb{E}[f(\mathbf{x})^2]=m=2^n.$ To compute $\mathbb{E}[f(\mathbf{x})^4]$
observe that when $p,q$ and $s$ are arbitrarily fixed we have $\mathbb{E}[f_p(\mathbf{x})f_q(\mathbf{x})f_s(\mathbf{x})f_t(\mathbf{x})]\neq 0$ for a unique (one in $2^n$) choice of $t$. Hence, $\mathbb{E}[f(\mathbf{x})^4]=m^4/2^n=2^{3n}.$ Thus, $\mathbb{E}[f(\mathbf{x})^4]/\mathbb{E}[f(\mathbf{x})^2]^2=2^n$.
Observe that $\rho=2^{n-1}$ and $2\rho+1- \frac{2\rho}{m}=2^n$ as well.


\section{Hypercontractive Inequality}

A {\em multiset} may contain multiple appearances of the same element. For multisets we will use the same notation as
for sets, but we will stress it when we deal with multisets. We do not attempt to optimize $g(r)$ in the following theorem.

\begin{theorem}\label{2r2in}
Let $f(x)$ be a pseudo-Boolean function of width $\rho\ge 1$. Then for each positive integer $r$ we have $||f||_{2r}\le [g(r){\rho}^{r-1}]^{\frac{1}{2r}}\cdot ||f||_2,$ where
$g(r) = (2r)!.$
\end{theorem}
\begin{proof}
Observe that
$\mathbb{E}[f(\mathbf{x})^{2r}]=\sum {2r \choose \alpha_1 \ldots \alpha_m}\mathbb{E}[f^{\alpha_1}_{1}(\mathbf{x})\cdots f^{\alpha_m}_{m}(\mathbf{x})],$ where the sum is taken over all partitions $\alpha_1 + \cdots + \alpha_m=2r$ of $2r$ into $m$ non-negatives summands.
Consider a non-zero term $\mathbb{E}[f^{\alpha_1}_{1}(\mathbf{x})\cdots f^{\alpha_m}_{m}(\mathbf{x})].$
Note that each variable $\mathbf{x}_i$ appears in an
even number of the factors in $f^{\alpha_1}_{1}(\mathbf{x})\cdots f^{\alpha_m}_{m}(\mathbf{x}).$
We denote the set of all such $m$-tuples $\alpha=(\alpha_1,\ldots ,\alpha_m)$ by $\cal E$. Then
\begin{equation}\label{xreq} \mathbb{E}[f(\mathbf{x})^{2r}]=  \sum_{\alpha\in {\cal E}} {2r \choose \alpha} \prod_{i=1}^mw^{\alpha_i}_i. \end{equation}

It is useful for us to view
$f^{\alpha_1}_{1}(\mathbf{x})\cdots f^{\alpha_m}_{m}(\mathbf{x})$, $\alpha\in \cal E$, as a product of $2r$ factors $f_{i}(\mathbf{x})$,
i.e., $$\mathbb{E}[f^{\alpha_1}_{1}(\mathbf{x})\cdots f^{\alpha_m}_{m}(\mathbf{x})]=
\mathbb{E}[f_{t_1}(\mathbf{x})\cdots f_{t_{2r}}(\mathbf{x})].$$

Let $I$ be a subset of the multiset $\{t_1,\ldots ,t_{2r}\}$ ($I$ is a multiset).
We call $I$ is {\em nontrivial} if it contains at least two elements (not necessarily distinct).
A subset $J$ of $I$ is called {\em minimally even} if $J$ is nontrivial,
$\mathbb{E}[\prod_{i\in J}f_{i}(\mathbf{x})]\neq 0$ but $\mathbb{E}[\prod_{i\in K}f_{i}(\mathbf{x})]= 0$
for each nontrivial subset $K$ of the multiset $J$. If $I_1=\emptyset$ (that is $\emptyset \in {\cal F}$) and
$1$ is an element of $I$ without repetition (i.e., only one copy of $1$ is in $I$), then $\{1\}$ is also called a
{\em minimally even} subset. (Thus, if $I$ contains two or more elements 1 then $\{1,1\}$ is minimally even,
but $\{1\}$ is not; however, if $I$ contains just one element 1, then $\{1\}$ is minimally even.)


Let $\mu_1$ be an element in the multiset $T_1:=\{t_1,\ldots ,t_{2r}\}$ such that
$w^2_{\mu_1}=\max\{w^2_{t_i}:\ t_i\in T_1\},$ and let $M_1$ be a minimally even subset of $T_1$ containing $\mu_1$. For $j\ge 2$, let
$\mu_j$ be an element in the multiset $T_j:=\{t_1,\ldots ,t_{2r}\}\setminus (\cup_{i=1}^{j-1}M_i)$ such that
$w^2_{\mu_j}=\max\{w^2_{t_i}:\ t_i\in T_j\},$ and let $M_j$ be a minimally even subset of $T_j$ containing $\mu_j$.
Let $s$ be the largest $j$ for which $\mu_j$ is defined above. Observe that $s \leq r$ as at most one of the
minimally even sets $M_1,M_2, \ldots, M_s$ has size one.
If $s<r$, for every $j\in \{s+1,s+2,\ldots ,r\}$ let
$\mu_j$ be an element in the multiset $T_1$ such that
$w^2_{\mu_j}=\max\{w^2_q:\ q\in T_1\setminus \{\mu_1,\ldots ,\mu_{j-1}\}\}.$

Let $\alpha\in \cal E$. For every $i\in [m]$, let $\beta_i=\beta_i(\alpha)$ be the number of copies of $i$ in the multiset $\{\mu_1,\ldots ,\mu_r\}.$
Let ${\cal E}':=\{\beta(\alpha):\ \alpha\in {\cal E}\}.$ The $2r$ terms in $\prod_{t\in T_1}w_t=\prod_{i=1}^m w_i^{\alpha_i}$ can be split into $r$
pairs such that each pair contains exactly one element with its index in the multiset
$\{\mu_1,\ldots,\mu_r\}$ and, furthermore, in each pair, the
element with its index in the multiset has at least as high an absolute value as the other element. Therefore the following holds.  \begin{equation}\label{alphabeta} \prod_{i=1}^mw_i^{\alpha_i}\le \prod_{i=1}^mw_i^{2\beta_i(\alpha)}.\end{equation}

For an $m$-tuple $\beta\in {\cal E}'$, let $N(\beta)$ be the number of $m$-tuples $\alpha\in \cal E$ such that $\beta=\beta(\alpha).$
We will now give an upper bound on $N(\beta)$, by showing how to construct
all possible $\alpha$ with $\beta(\alpha)=\beta$. Let
$M=\{\mu_1,\ldots,\mu_r\}$ be the multiset containing $\beta_i$ copies of $i$.
We first partition $M$ into any number of non-empty subsets. This can be done
in at most $r!$ ways, since we can place $\mu_1$ in the ``first'' subset, $\mu_2$ in the same subset or in the ``second'' subset, etc.
Each of the subsets will be a subset of a minimal even multiset. Thus, while any multiset, $M_i'$, is not a minimally even subset,
there is an $\mathbf{x}_j$ of odd total degree in $\prod_{t\in M'_i}f_{t}(\mathbf{x})$.
Thus, to construct a minimally even subset from $M_i'$,
we have to add to $M'_i$ an element $q$ such that $f_q(\mathbf{x})$ contains $\mathbf{x}_j$,
which restricts $q$ to at most $\rho$ choices. Continuing in this manner,
observe that we have at most $\rho$ choices for the $r$ extra elements we
need to add. As the very last element we add has to be unique we note that we
construct at most $r! \rho^{r-1}$ partitions of $T_1$ into minimally even subsets
in this way. For each such partition, we have $\alpha=(\alpha_1,\ldots,\alpha_m)$, where
$\alpha_i$ is the number of occurrences of $i$ in $T_1$. Note that
every $\alpha$ for which $\beta(\alpha)=\beta$ can be constructed this way,
which implies that \begin{equation}\label{Neq} N(\beta)\le \rho^{r-1}r!.\end{equation}

Let $\alpha \in \cal E$ and $\beta(\alpha)=(\beta_1,\ldots ,\beta_m)$. By the construction of $\beta(\alpha),$
each non-zero $\beta_i$ appears in the multiset $\{\beta_1,\ldots ,\beta_m\}$ at least as many times as in $\{\alpha_1,\ldots ,\alpha_m\}$.
This implies that \begin{equation}\label{diveq}{2r \choose \alpha}/{r \choose \beta(\alpha)}
\le (2r)!/r!.\end{equation}

By Parseval's Identity, \begin{equation}\label{pareq}\mathbb{E}[f(\mathbf{x})^{2}]^r=\left(\sum_{i=1}^m w_i^2\right)^r = \sum {r \choose b_1 \ldots b_m}w_1^{2b_1}\cdots w_m^{2b_m},\end{equation} where the last sum is taken over all partitions $b_1 + \cdots + b_m=r$ of $r$ into $m$ non-negatives integral summands.

Now by (\ref{xreq}), (\ref{alphabeta}), (\ref{Neq}), (\ref{diveq}) and (\ref{pareq}), we have

\begin{center}
$\begin{array}{rcl}
\mathbb{E}[f(\mathbf{x})^{2r}] &= & \sum_{\alpha\in {\cal E}} {2r \choose \alpha} \prod_{i=1}^mw^{\alpha_i}_i\\
    &\le & \sum_{\alpha\in {\cal E}} {2r \choose \alpha} ({r \choose \beta(
    \alpha)}/{r \choose \beta(\alpha)}) \prod_{i=1}^mw^{2\beta_i(\alpha)}_i\\
  & \le & \sum_{\beta\in {\cal E}'} N(\beta)((2r)!/r!){r \choose \beta}\prod_{i=1}^mw^{2\beta_i}_i \\
  & \leq & (2r)!\rho^{r-1}\sum_{\beta\in {\cal E}'}{r \choose \beta}\prod_{i=1}^mw^{2\beta_i}_i \\
  & \le &  (2r)!\rho^{r-1} \mathbb{E}[f(\mathbf{x})^{2}]^r.\\
 \end{array}$
\end{center}

\end{proof}
We can get a better bound on $N(\beta)$ in the proof of this theorem as follows. Note that the number of partitions of a set
of cardinality $r$ into non-empty subsets is called the $r$th Bell number, $B_r$, and there is an upper bound on $B_r$: $B_r < \left(\frac{0.792r}{\ln(r+1)}\right)^r$
\cite{BerTas2010}. This upper bound is better than the crude one, $B_r \le r!$, that we used in the proof of this theorem, but our bound allowed us to obtain a simple expression for $g(r)$. Moreover, we believe that the following, much stronger, inequality holds.

\begin{conjecture}\label{conj}
There exists a constant $c$ such that for every
pseudo-Boolean function $f(x)$ of width $\rho\ge 1$ we have $||f||_{2r}\le c\sqrt{r\rho}||f||_2$ for each positive integer $r$.
\end{conjecture}

If Conjecture \ref{conj} holds then it would be best possible, in a sense, due to the following example.
Let $f(x)=\sum_{i=1}^nx_i$. By Parseval's Indentity, $\mathbb{E}[f(\mathbf{x})^2]=n.$ We will now give a bound for
$\mathbb{E}[f(\mathbf{x})^{2r}]$. Define $(a_1,a_2,\ldots,a_{2r})$ to be a {\em good} vector if all $a_i$ belong to
$[n]=\{1,2,\ldots,n\}$ and any number from $[n]$ appears in the vector zero times or exactly twice.
The number of good vectors is equal to ${n \choose r} \frac{(2r!)}{2^r}$, which implies that
$  \mathbb{E}[f(\mathbf{x})^{2r}]  \geq  {n \choose r} \frac{(2r!)}{2^r}
  =  \frac{n!}{(n-r)!} \times \frac{(2r)!}{2^r r!}$.

Note that $\frac{(2r)!}{2^r r!}=(2r-1)!!>(r/e)^r$ and, when $n$ tends to infinity, $ \frac{n!}{(n-r)!}$ tends to $n^r = \mathbb{E}[f(\mathbf{x})^2]^r$. Therefore, the bound in Conjecture \ref{conj} (for $\rho=1$) cannot be less than $c \sqrt{r}$ for some constant $c$.


\vspace{3mm}

Theorem \ref{2r2in} can be easily extended as follows.

\begin{corollary}\label{cor}
Let $f(x)$ be a pseudo-Boolean function of width $\rho\ge 1$. Then for each $q> p\ge 2$ we have $ ||f||_q\le ((2r)!\rho^{r-1})^{1/(2r)}||f||_p,$ where $r=\lceil q/2\rceil$.
\end{corollary}
\begin{proof}
Let $r=\lceil q/2\rceil$. Using Theorem \ref{2r2in} and the fact that $||f||_s\ge ||f||_t$ for each $s>t>1$, we obtain $$||f||_q\le ||f||_{2r}\le ((2r)!\rho^{r-1})^{1/(2r)}||f||_2\le ((2r)!\rho^{r-1})^{1/(2r)}||f||_p.$$
\end{proof}

\section{Application of Theorem \ref{42in}}\label{sec:fr}

Consider the following problem {\sc MaxLin-AA} first
studied in the literature on approximation algorithms, cf. \cite{Hastad01,HasVen04}. H\aa stad  \cite{Hastad01} succinctly summarized the importance of the maximization version of this problem by saying that it is ``in many respects as basic as satisfiability.''
We are given a nonnegative integer $k$ and a system $S$ of equations $\prod_{i \in I_j}x_i = b_j$, where $x_i, b_j \in \{-1,1\}$, $j=1,\ldots ,m$ and where each
equation is assigned a positive integral weight $w_j$.
The question is whether there is an assignment of values $\{-1,1\}$ to the variables $x_i$ such that
the total weight of satisfied equations is at least $W/2+k$, where $W$ is the total weight of all equations.
If we assign values randomly, the expected weight of satisfied equations is $W/2$ and, thus, $W/2$ is a lower bound on the total weight of satisfied equations.
Hereafter, we assume that no two equations of $S$ have the same left-hand side.

Mahajan {\em et al.} \cite{MahajanRamanSikdar09} asked whether {\sc MaxLin-AA} is fixed-parameter tractable with respect to the parameter $k$, i.e.,  whether\footnote{This is a definition equivalent to the one usually used in the area of parameterized algorithms and complexity.
For more information on the area, see, e.g., \cite{DowneyFellows99,FlumGrohe06}.} there exists a function $h(k)$ in $k$ only and a polynomial time algorithm that transforms $S$ into a new system $S'$ with $m'$ equations and $n'$ variables, and parameter $k'$ such that $n'm'+k'\le h(k)$ and we can satisfy equations of $S$ of total weight at least $W/2+k$ if and only if we can satisfy equations of $S'$ of total weight at least $W'/2+k'.$ Here $W'$ is the total weight of all equations in $S'$. This question was answered in affirmative in a series of two papers \cite{CroGutJonKimRuz10,CroFelGutJonRosThoYeo}, where an exponential function $h(k)$ was obtained. The authors of \cite{CroFelGutJonRosThoYeo} asked whether the result can be strengthen to $h(k)$ being a polynomial (this is a natural question in the area of parameterized algorithms and complexity due to applications in preprocessing).

It was proved in \cite{GutKimSzeYeo11} that $h(k)=O(k^4)$ when (i) every equation has an odd number of variables, or (ii) no equation has more than $r$ variables, where $r$ is a constant, or (iii) no variable appears in more than $\rho$ equations, where $\rho$ is a constant. Cases (ii) and (iii) can be extended to $r$ and $\rho$ being functions of $n$ and $m$, respectively. Below we consider (iii) in some detail. Case (ii) can be treated in a similar way using (\ref{hyperc}).

Note that the answer to {\sc MaxLin-AA} is {\sc Yes} if and only if the maximum of polynomial $Q=\sum_{j=1}^m c_j\prod_{i\in I_j}x_i$ is at least $2k$,  where $c_j=w_{j}b_j$ and each $x_i\in \{-1,1\}$.
Assign $-1$ or $1$ to each variable $x_i$ independently and uniformly at random. Then $Q$ is a random variable.

We will use the following lemma of Alon {\em et al.} \cite{AloGutKimSzeYeo11}: Let $X$ be a real random variable and suppose that its first, second and
fourth moments satisfy $\Exp[X]=0$ and $\Exp[X^4]
\leq b \Exp[X^2]^2$, where $b$ is a positive constant. Then
$\Prob(X \geq \frac{1}{2} \sqrt{\Exp[X^2]/b})>0.$

Observe that $\Exp[Q]=0$ and $\Exp[Q^2]=\sum_{j=1}^m c^2_j.$ By Theorem \ref{42in}, $\Exp[Q^4]
\leq (2\rho + 1) \Exp[Q^2]^2$ and, thus, $\Prob(Q \geq \frac{1}{2} \sqrt{\sum_{j=1}^m c^2_j/(2\rho + 1)})>0.$ Since $\sum_{j=1}^m c^2_j\ge m$ we have
 $\Prob(Q \geq\frac{1}{2} \sqrt{m/(2\rho + 1)})>0.$ Thus, if $\frac{1}{2} \sqrt{m/(2\rho + 1)}\ge 2k$ the answer to {\sc MaxLin-AA} is {\sc Yes}. Otherwise,
 $m\le 8 (2\rho + 1)k^2$ and so $m$ is bounded by a polynomial in $k$ if $\rho\le m^{\alpha}$ for some constant $\alpha<1$.
 It is shown in \cite{GutKimSzeYeo11} that we may assume that $n\le m$ as otherwise we can replace $S$ by an equivalent system for which $n\le m$ holds.
This implies that $nm$ is bounded by a polynomial in $k$ if $\rho\le m^{\alpha}$ for some constant $\alpha<1$.

Now assume that $\rho\le m^{\alpha}$ for some constant $\alpha<1$. To construct the required system $S'$,
check whether $\frac{1}{2} \sqrt{m/(\rho + 1)}\ge 2k$. If the answer is {\sc Yes}, let $S'$ be an arbitrary consistent system of $2k$  equations with all weights equal 1 and, otherwise, $S'=S$. The parameter $k'=k$.

Note that the bound $\rho\le m^{\alpha}$ is only possible because the coefficient in Theorem \ref{42in} is so small.

\section{Further Research}

It would be interesting to verify Conjecture \ref{conj} and decrease the coefficient before $||f||_p$ in  Corollary \ref{cor}.

\paragraph{Acknowledgments}
This research was partially supported by an International Joint grant of Royal Society. Part of the paper was written when the first author was attending Discrete Analysis programme of the Isaac Newton Institute for Mathematical Sciences, Cambridge. Financial support of the Institute is greatly appreciated.
We are thankful to Franck Barthe and Hamed Hatami for useful discussions on the paper.

\end{document}